\documentclass[sigconf,nonacm,balance=false]{acmart}

\usepackage{amsmath}
\usepackage{multirow}
\usepackage{array}
\usepackage{xspace}
\newcommand{\method}{\textsc{CC-MGRec}\xspace}
\newcommand{\rarc}{\textsc{RGSC}\xspace}
\newcommand{\cosig}{\textsc{CC-MG}\xspace}
\newcommand{\cscp}{\textsc{CSCP}\xspace}
\newcommand{\R}{\mathbb{R}}
\newcommand{\U}{\mathcal{U}}
\newcommand{\I}{\mathcal{I}}
\newcommand{\Ec}{\mathcal{E}_{c}}
\newcommand{\Es}{\mathcal{E}_{s}}

\setcopyright{none}
\acmDOI{}
\acmISBN{}

\begin{document}

\title{Rethinking Evidence Modeling: Context-Coupled Graph Learning for Multimodal Recommendation}

\author{Mingqiao Zhang}
\affiliation{%
  \institution{Nanjing University}
  \city{Nanjing}
  \country{China}}
\email{mingqiaozhang@smail.nju.edu.cn}

\author{Qiyao Peng}
\authornote{Corresponding author.}
\affiliation{%
  \institution{Tianjin University}
  \city{Tianjin}
  \country{China}}
\email{qypeng@tju.edu.cn}

\author{Yumeng Wang}
\affiliation{%
  \institution{Tianjin University}
  \city{Tianjin}
  \country{China}}
\email{ymwang@tju.edu.cn}

\author{Wenzhe Dong}
\affiliation{%
  \institution{Tianjin University}
  \city{Tianjin}
  \country{China}}
\email{dongwz@tju.edu.cn}

\author{Hongtao Liu}
\affiliation{%
  \institution{Tianjin University}
  \city{Tianjin}
  \country{China}}
\email{htliu@tju.edu.cn}

\author{Jian Yang}
\affiliation{%
  \institution{Beihang University}
  \city{Beijing}
  \country{China}}
\email{jiaya@buaa.edu.cn}

\renewcommand{\shortauthors}{Mingqiao Zhang et al.}

\begin{abstract}
Multimodal graph recommendation integrates observed interactions
with item images and text to learn user preferences. These inputs
constitute Semantic--Behavioral Evidence (SBE).
Many existing graph methods incorporate collaborative and
multimodal semantic relations, yet their aggregation and state
updates remain context-separated, limiting the continuous
composition of preference evidence across layers and the joint
use of multiple sources.
To address this limitation, we propose \method, a context-coupled
graph learning method that organizes SBE within a common graph
learning context for cross-relation modeling and compositional
propagation.
Its context-coupled multimodal graph integrates user--item
collaborative relations and item--item semantic relations,
coordinating their message scales through joint normalization.
Recommendation-grounded semantic calibration uses training
behavior to adjust semantic neighbor weights while preserving
normalized collaborative propagation coefficients.
Collaborative--semantic compositional propagation updates shared
node states layer by layer to continuously compose collaborative
and multimodal semantic evidence, implicitly forming Preference
Evidence Chains (PECs), including collaborative--semantic and
visual--textual PECs.
Joint execution of relation messages on the same shared node
states within each layer preserves propagation outputs while
reducing intermediate output processing and fusion overhead.
We provide theoretical analysis of the proposed mechanisms.
Experiments on three Amazon datasets demonstrate improved
recommendation quality and faster inference during graph propagation.
\end{abstract}



\maketitle

\section{Introduction}
\begingroup
\setlength{\emergencystretch}{1.5em}

Multimodal graph recommendation expresses SBE as node information
and graph relations, and forms preference representations through
propagation and fusion \citep{xu2025cohesion}.
At the representation level, VBPR incorporates visual content into
preference representations, MMGCN learns modality-specific
representations on the interaction graph, and SG-URInit constructs
model-agnostic user initializations from historical item features
and cluster information
\citep{he2016vbpr,wei2019mmgcn,xu2026well}.
UGT jointly optimizes image--text feature extraction and
recommendation learning, aggregating multimodal content and ID
information on the interaction graph \citep{yi2024unified}. Evidence modeling also organizes graph relations and state updates. GRCN refines observed interactions using content--preference agreement; LATTICE and FREEDOM construct content-derived item
graphs, while DualGNN introduces user associations
\citep{wei2020graph,zhang2021mining,zhou2023tale,wang2021dualgnn}.
In rating-feedback settings, SiMGR represents interactions through
positive and negative relation graphs \citep{lian2026sign}.
For state updates, MIG-GT maintains modality-specific node states
and receptive fields on the same interaction graph, whereas DHMRec
passes interaction-graph states into subsequent item-graph learning
\citep{hu2025modality,zhan2026dhmrec}.

These studies enrich evidence content, associations, and fusion
strategies. Yet does expressing diverse evidence as node information
and graph relations suffice for evidence carried by different
relations to compose continuously across layers and jointly support
preference learning?
To answer this question, we examine the learning environment
formed by node information, weighted relation neighborhoods,
normalization, and state updates, which we term the
\emph{graph learning context}
\citep{he2020lightgcn}.
Among representative designs using both collaborative and
item--item semantic relations, fusion, alignment, or stage-wise
transfer connects modality-specific or relation-specific state
streams. The two relation types do not jointly update the same
evolving node states at every layer.
We refer to this organization as \emph{context-separated}.

Figure~\ref{fig:evidence-contexts} contrasts learning from the same SBE under separate and coupled contexts. The illustrated separate context maintains relation-specific node states.
In a \emph{context-coupled} organization, collaborative and semantic relations jointly determine message scales through joint
normalization and update shared node states at every layer.
Evidence aggregated through one relation can therefore continue
through other relations in subsequent layers.
We use \emph{Preference Evidence Chains} (PECs) to denote ordered, weighted propagation contributions implicitly formed by multilayer
graph learning and incorporated into recommendation representations.
Relative to the illustrated separate context, shared node-state
updates retain within-relation PECs and support two additional
compositions.
\emph{Collaborative--semantic PECs} (C--S PECs) allow interaction
evidence to continue through item semantic relations and semantic
evidence to re-enter collaborative learning.
\emph{Visual--textual PECs} (V--T PECs) allow evidence aggregated
through one modality's relations to continue through the other,
composing the item connections revealed by both modalities.

Two preliminary diagnostics examine the recommendation value of
these compositions.
C--S PEC support formed from training interactions and the
pre-calibration semantic graph covers 10.84\% of held-out positives
on Baby versus 1.32\% of matched unseen controls
(Figure~\ref{fig:fragmentation-evidence}(a)).
In fixed trained models, removing V--T PEC contributions while
retaining constituent evidence, relation support, and single-modality
paths produces the largest ranking losses in high-exposure groups;
some low-exposure groups change little or improve
(Figure~\ref{fig:fragmentation-evidence}(b)).
These support-coverage differences and within-model ranking effects provide
preliminary empirical motivation for supporting useful evidence
compositions.

\begin{figure}[t]
  \centering
  \scalebox{1}[1.15]{%
    \includegraphics[width=\columnwidth]{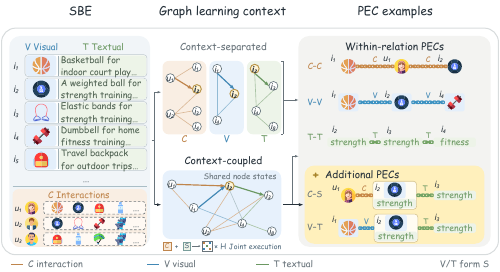}%
  }
  \caption{From SBE to PECs under context-separated and context-coupled
  graph learning. Gold highlights the additional C--S and V--T PECs
  relative to the illustrated separate context.}
  \label{fig:evidence-contexts}
  \Description{The schematic shows paired visual and textual
  item evidence and user interaction histories. Context-separated graph
  learning maintains distinct node states, whereas context-coupled
  graph learning shares states across collaborative and semantic relations.
  Five example PECs use chain links to trace visual,
  textual, and interaction SBE, highlighting additional
  collaborative--semantic and visual--textual PECs. The coupled
  context supports joint execution with less intermediate overhead.}
\end{figure}

\begin{figure}[t]
  \centering
  \includegraphics[width=\columnwidth]{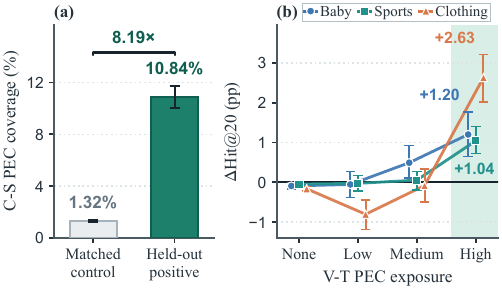}
  \caption{Preliminary diagnostics of PEC composition.
  (a) C--S PEC support coverage of held-out positives and matched
  unseen controls on Baby, using training interactions and the
  pre-calibration semantic graph.
  (b) Hit@20 differences (with minus without V--T PEC contributions)
  in fixed models, grouped by pre-removal exposure.
  Error bars indicate user-clustered 95\% confidence intervals.}
  \label{fig:fragmentation-evidence}
  \Description{The first panel compares collaborative--semantic PEC
  support coverage of held-out positives and matched unseen controls
  on Baby. The second panel shows Hit@20 with minus without
  visual--textual PEC contributions in fixed models, grouped by
  pre-removal exposure. Error bars show user-clustered 95 percent
  confidence intervals.}
\end{figure}

Motivated by these observations, we propose
\emph{Context-Coupled Graph Learning}, implemented by \method,
which organizes SBE into a common graph learning context through
graph construction, semantic calibration, and compositional
propagation.
\textbf{Context-Coupled Multimodal Graph Construction} (\cosig)
organizes node information and weighted relations from multimodal
content and training behavior, jointly normalizing collaborative
and semantic relations.
Under joint normalization, semantic reweighting can also alter
collaborative propagation coefficients.
\textbf{Recommendation-Grounded Semantic Calibration} (\rarc)
reallocates semantic neighbor weights using training behavior while
preserving normalized collaborative propagation coefficients,
with collaborative relations and other construction settings fixed. \textbf{Collaborative--Semantic Compositional Propagation} (\cscp)
then updates shared node states layer by layer to implicitly form and aggregate C--S and V--T PECs.
Shared states also support joint execution of relation messages.
For identical relation operators and current-layer input states,
a single assembled sparse operator produces outputs equivalent to
relation-wise computation followed by summation, while reducing
intermediate output accesses, fusion operations, and operator calls
during graph-propagation inference.

On Amazon Baby, Sports, and Clothing, \method achieves the highest
mean Recall and NDCG among the compared models and faster
graph-propagation inference than MIG-GT.
Our contributions are:
\begin{itemize}
\item We propose \method, a context-coupled graph learning framework
that uses joint normalization and shared node-state updates to
continuously compose SBE across layers, with C--S and V--T PECs
characterizing cross-relation and cross-modal propagation
contributions.
\item We develop recommendation-grounded semantic calibration that
preserves normalized collaborative propagation coefficients, and
output-equivalent joint execution that reduces relation-level
intermediate output accesses and fusion overhead.
\item We analyze evidence composition, calibration invariance, and
execution equivalence, and evaluate recommendation quality and
graph-propagation inference efficiency on three datasets.
Preliminary diagnostics, retrained mechanism controls, and
output-equivalent execution experiments test the corresponding
mechanisms.
\end{itemize}
\endgroup

\section{Related Work}
\suppressfloats[t]

\subsection{Multimodal Recommendation}

Multimodal recommendation extracts complementary preference evidence
from item content and user behavior \citep{liu2024multimodal}.
Images and text describe item appearance, attributes, and usage,
observed interactions reflect behavioral preferences~\citep{liu2021toward}, and behavioral
co-occurrence reveals user or item associations supported by common
interactions. VBPR and ACF incorporate content into preference
representations \citep{he2016vbpr,chen2017attentive}.
SLMRec, MMSSL, and BM3 enhance multimodal representations through
self-supervised learning
\citep{tao2022self,wei2023multi,zhou2023bootstrap}, while SMORE improves
representation quality through spectrum-based denoising and modality
fusion \citep{ong2025spectrum}.
Evidence utilization also extends to input augmentation,
initialization, and behavioral association modeling.
SG-URInit constructs user initializations from historical item features
and group information \citep{xu2026well}; MLLMRec enhances preference
reasoning with multimodal large language models
\citep{dang2026mllmrec}; GRE-MC completes missing modalities through
graph retrieval \citep{li2026robust}.
DualGNN uses common interactions to establish user associations
\citep{wang2021dualgnn}, IGDMRec incorporates item co-occurrence
information to refine semantic relations \citep{guo2026igdmrec}, and
BRIDGE uses behavior to calibrate candidate scores
\citep{li2026behavior}.
These studies enrich the evidence used for preference learning and
the ways in which it is exploited \citep{peng2025survey}.

\subsection{Graph Learning for Recommendation}

Graph learning organizes evidence as nodes and relations, modeling
higher-order associations through neighborhood propagation \citep{wang2019neural}.
Table~\ref{tab:related_taxonomy} compares representative methods
across evidence expression, graph construction and propagation,
and PEC support.

\begin{table}[!t]
\centering
\caption{Comparison of evidence modeling and PEC support.
\textnormal{\textbf{NI}: modality content in initial node states;
\textbf{CO}: explicit behavioral co-occurrence statistics;
\textbf{SR}: item semantic relations;
\textbf{BW}: semantic weights conditioned on behavioral inputs;
\textbf{JN}: joint normalization of observed interactions and semantic relations;
\textbf{SS}: shared node updates across both relation types at every layer.
\textbf{C--S/V--T}: collaborative--semantic/visual--textual PECs.
Architectural support assumes valid paths and sufficient propagation depth.}}
\label{tab:related_taxonomy}
\begingroup
\footnotesize
\setlength{\tabcolsep}{3pt}
\renewcommand{\arraystretch}{1.15}
\def\rwyes{\ding{51}}
\def\rwno{\ensuremath{\times}}

\begin{tabular*}{\columnwidth}
  {@{\extracolsep{\fill}}lcccccccc@{}}
\toprule
& \multicolumn{2}{c}{\textbf{\shortstack{Evidence\\expression}}}
& \multicolumn{4}{c}{\textbf{\shortstack{Graph construction\\and propagation}}}
& \multicolumn{2}{c}{\textbf{\shortstack{PEC\\support}}} \\
\cmidrule(lr){2-3}
\cmidrule(lr){4-7}
\cmidrule(lr){8-9}
\textbf{Method}
& \textbf{NI} & \textbf{CO}
& \textbf{SR} & \textbf{BW} & \textbf{JN} & \textbf{SS}
& \textbf{C--S} & \textbf{V--T} \\
\midrule
MMGCN   & \rwyes & \rwno  & \rwno  & \rwno  & \rwno  & \rwno  & \rwno  & \rwno  \\
GRCN    & \rwyes & \rwno  & \rwno  & \rwno  & \rwno  & \rwno  & \rwno  & \rwno  \\
DualGNN & \rwyes & \rwyes & \rwno  & \rwno  & \rwno  & \rwno  & \rwno  & \rwno  \\
MIG-GT  & \rwyes & \rwno  & \rwno  & \rwno  & \rwno  & \rwno  & \rwno  & \rwno  \\
UGT     & \rwyes & \rwno  & \rwno  & \rwno  & \rwno  & \rwno  & \rwno  & \rwno  \\
\addlinespace[2pt]
LATTICE & \rwno  & \rwno  & \rwyes & \rwno  & \rwno  & \rwno  & \rwno  & \rwyes \\
MICRO   & \rwno  & \rwno  & \rwyes & \rwno  & \rwno  & \rwno  & \rwno  & \rwno  \\
FREEDOM & \rwno  & \rwno  & \rwyes & \rwno  & \rwno  & \rwno  & \rwno  & \rwyes \\
MGCN    & \rwyes & \rwno  & \rwyes & \rwno  & \rwno  & \rwno  & \rwyes & \rwno  \\
DHMRec  & \rwyes & \rwno  & \rwyes & \rwno  & \rwno  & \rwno  & \rwyes & \rwyes \\
SCALE   & \rwno  & \rwno  & \rwyes & \rwno  & \rwno  & \rwno  & \rwno  & \rwno  \\
IGDMRec & \rwno  & \rwyes & \rwyes & \rwyes & \rwno  & \rwno  & \rwno  & \rwyes \\
\midrule
\textbf{\method}
        & \rwyes & \rwyes & \rwyes & \rwyes & \rwyes & \rwyes & \rwyes & \rwyes \\
\bottomrule
\end{tabular*}

\endgroup
\end{table}

EgoGCN organizes multimodal messages through edge-wise modulation,
while GRCN refines observed interaction relations
\citep{chen2022breaking,wei2020graph}.
LATTICE, MICRO, FREEDOM, and MGCN introduce item semantic relations
\citep{zhang2021mining,zhang2022latent,zhou2023tale,yu2023multi};
LGMRec and DOGE further incorporate higher-order structures
\citep{guo2024lgmrec,meng2025doge}.
In rating-feedback settings, SiMGR distinguishes positive and negative
interaction relations \citep{lian2026sign}.
Propagation organization determines how evidence carried by different
relations participates in state updates.
MMGCN and MIG-GT use modality-specific propagation
\citep{wei2019mmgcn,hu2025modality}, whereas CRANE, MGCN, and DHMRec
connect collaborative and semantic learning through staged updates
\citep{dai2026cross,yu2023multi,zhan2026dhmrec}.
UGT unifies multimodal feature learning on the interaction graph
\citep{yi2024unified}, while SCALE learns collaborative and semantic
representations separately and connects them through alignment and
fusion \citep{li2026subspace}.
These studies organize evidence through different graph construction
and propagation choices.
From the perspective of PECs, existing structures support different
forms of evidence composition.
MGCN and DHMRec support C--S PECs through cross-relation state transfer
\citep{yu2023multi,zhan2026dhmrec}, while the merged semantic relations
in LATTICE and FREEDOM support V--T PECs at sufficient propagation
depth \citep{zhang2021mining,zhou2023tale}.
Graph learning research also addresses propagation efficiency \citep{besta2024parallel}.
HAGs reduce redundant computation through aggregation reuse
\citep{jia2020redundancy}, while Hector and TLV-HGNN optimize
relational graph computation \citep{wu2024hector,han2025tlv}.
Our \method organizes multimodal content, observed interactions,
and behavioral co-occurrence information within a common graph
learning context to support continuous evidence composition across
relations and efficient inference.

\section{Problem Formulation}

Let $\U=\{u\}$ and $\I=\{i\}$ denote the user and item sets, with
$M=|\U|$ and $N=|\I|$.  Training interactions are represented by
$\mathbf R\in\{0,1\}^{M\times N}$, where $R_{ui}=1$ indicates an
observed interaction.  We write
$\mathcal N_u=\{i\in\I\mid R_{ui}=1\}$ for the training history of
user $u$, $\mathcal U_i=\{u\in\U\mid R_{ui}=1\}$ for the users
associated with item $i$, and $d_i=|\mathcal U_i|$ for its interaction
degree.  Each item has a visual feature
$\mathbf x_i^v\in\R^{d_v}$ and a textual feature
$\mathbf x_i^t\in\R^{d_t}$ extracted by pretrained models.

We distinguish observed user--item interactions, denoted by $\Ec$,
from content-supported item--item relations, denoted by $\Es$.  Both
are defined over $\mathcal V=\U\cup\I$.  For a directed weighted
adjacency $\mathbf A$, we use the destination-by-source convention:
$A_{pq}$ is the weight of the message from source $q$ to destination
$p$, and propagation is written as $\mathbf A\mathbf H$. We use
$[x]_+=\max(x,0)$, sigmoid $\sigma$, and $\operatorname{norm}$
for vector $\ell_2$ normalization. Given $\mathbf R$ and the item
features, our goal is to rank unobserved items for each user.

\section{Methodology}

\subsection{Framework Overview}

CC-MG is the core of \method: it organizes multimodal node
information and weighted collaborative and semantic relations
in a shared graph learning context (Figure~\ref{fig:framework}).
RGSC refines semantic neighbor weights using training behavior
while preserving their incoming budgets. CC-MG combines these
calibrated connections with semantic node anchors and weighted
interactions, then jointly normalizes the resulting graph.
CSCP learns from this graph through one evolving node state,
enabling cross-relation composition in a single propagation stream.
We first introduce the calibration rule used by CC-MG, followed
by graph construction, propagation, and recommendation learning.

\begin{figure*}[!t]
  \centering
  \includegraphics[width=\textwidth]{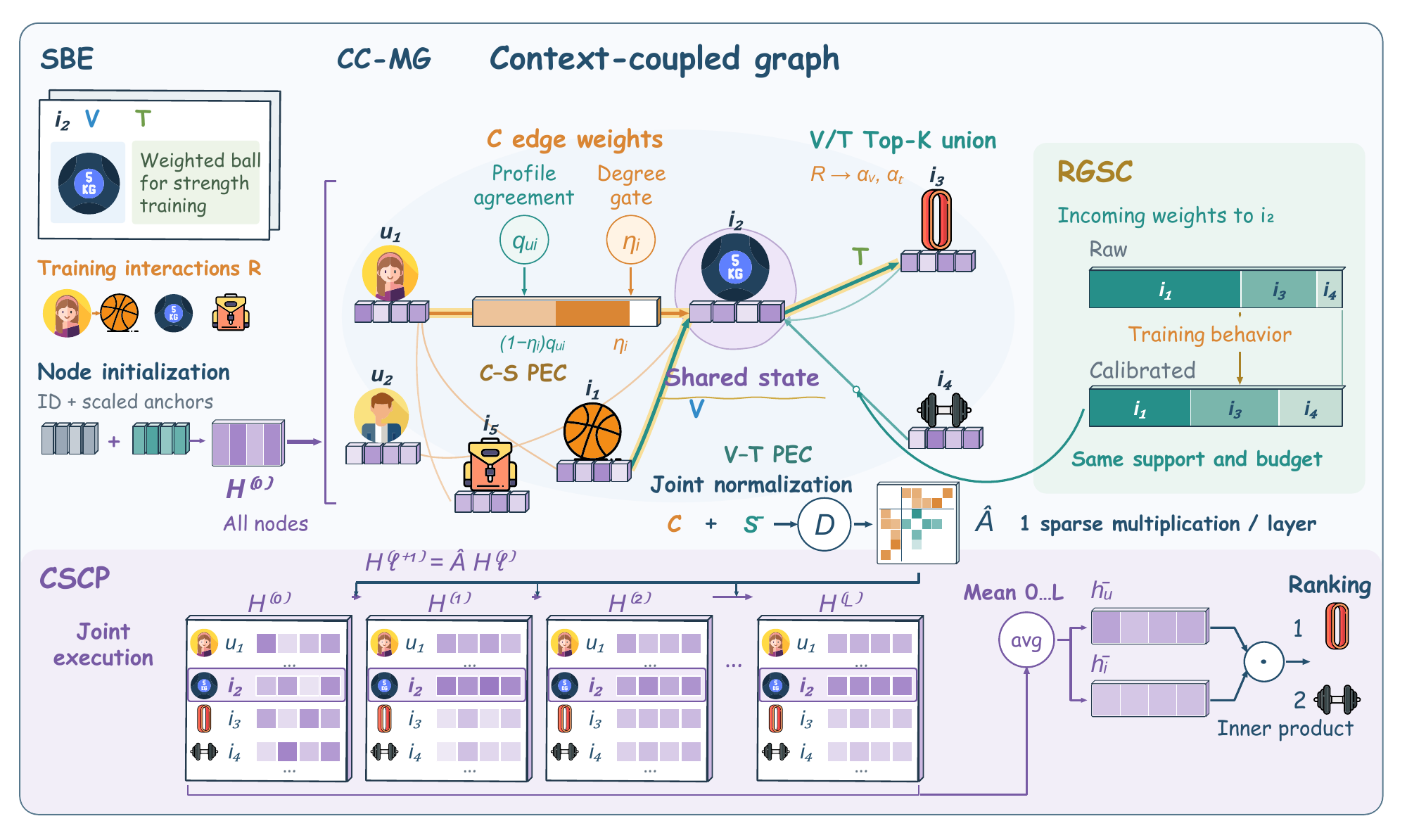}
  \caption{Overview of \method. CC-MG (top) organizes SBE into node
  information and weighted collaborative and semantic relations in
  a jointly normalized context-coupled graph. RGSC (right) reallocates
  semantic neighbor weights while preserving support and each item's
  incoming budget. CSCP (bottom) jointly updates a shared node state
  with one sparse operator, then averages all propagation depths for
  inner-product recommendation.}
  \label{fig:framework}
  \Description{The upper region shows SBE, node initialization, and
  a context-coupled graph with weighted collaborative and semantic
  relations. The right RGSC inset reallocates incoming semantic
  weights with unchanged support and budget. Joint normalization
  produces one shared propagation operator. The lower CSCP region
  shows selected rows of the full-node state across propagation
  depths, followed by depth averaging and inner-product ranking.}
\end{figure*}

\subsection{Recommendation-Grounded Semantic Calibration}
\label{sec:rgsc}

RGSC aligns semantic neighbor strengths with training behavior.
Given the candidate support $\Es$ and initial weights
$w^s_{i\rightarrow j}$ from CC-MG, it reallocates neighbor weights
while preserving each destination's incoming semantic budget.
The resulting semantic block $\bar{\mathbf S}$ supplies
behaviorally calibrated connections to the shared graph.

\paragraph{Behavioral relevance.}
Context--target co-occurrences measure whether a semantic neighbor
is more relevant than its popularity alone would suggest. We divide
each eligible user history into $F_u=\min(F,|\mathcal N_u|)$ folds,
using each fold once as target and the others as context. For a
supported direction $i\rightarrow j$, let $n_{i\rightarrow j}$
count context--target co-occurrences, $\nu_i$ count the corresponding
target opportunities, and $p_j=d_j/\sum_k d_k$ be the popularity
prior. Smoothing toward this prior with $\mu>0$ gives the log lift
\begin{equation}
\label{eq:behavioral-lift}
\ell_{i\rightarrow j}
=\log\frac{n_{i\rightarrow j}+\mu p_j}{(\nu_i+\mu)p_j}.
\end{equation}

A bounded relevance score combines both directions and discounts
limited behavioral evidence. Define the average lift
$\bar\ell_{ij}=(\ell_{i\rightarrow j}+\ell_{j\rightarrow i})/2$
and endpoint confidence $r_i=\nu_i/(\nu_i+\mu)$.
The pair confidence $\chi_{ij}=\sqrt{r_i r_j}$ then scales
the transformed lift:
\begin{equation}
\label{eq:behavioral-relevance}
\phi_{ij}=\chi_{ij}\tanh\!\left(
\frac{\sigma(\bar\ell_{ij})-1/2}{\tau_r}\right).
\end{equation}
Here $\tau_r>0$ controls sensitivity to the averaged lift.
Positive scores favor a neighbor, while negative scores
reduce its weight through the adjustment below.

\paragraph{Budget-preserving reallocation.}
Calibration first adjusts individual weights, then restores each
destination's original semantic budget. With strength $\gamma\geq0$,
the intermediate weight is
\begin{equation}
\label{eq:semantic-tilt}
\widetilde w^s_{i\rightarrow j}
=w^s_{i\rightarrow j}\exp(\gamma\phi_{ij}).
\end{equation}
Let $\mathcal B_j^s=\sum_{k:(k,j)\in\Es}w^s_{k\rightarrow j}$
denote the raw incoming budget. Rescaling the intermediate
weights gives
\begin{equation}
\label{eq:budget-recovery}
\bar w^s_{i\rightarrow j}
=\mathcal B_j^s
\frac{\widetilde w^s_{i\rightarrow j}}
{\sum_{k:(k,j)\in\Es}\widetilde w^s_{k\rightarrow j}}.
\end{equation}
This update preserves the nonzero support and satisfies
$\sum_i\bar w^s_{i\rightarrow j}=\mathcal B_j^s$.
Zero-budget destinations retain zero weights. The calibrated semantic block
$\bar S_{ji}=\bar w^s_{i\rightarrow j}$ is zero outside $\Es$.

\subsection{Context-Coupled Multimodal Graph Construction}
\label{sec:ccmg}

CC-MG organizes node information, relation connectivity, and
relation strength for learning in one propagation stream.
Visual and textual neighborhoods supply semantic connections;
behavioral agreement determines their modality contributions,
and RGSC calibrates individual neighbor weights.
Semantic anchors provide node information and also guide
interaction weighting. The weighted relations are assembled
and jointly normalized into a common propagation operator.
Together, these choices determine what recommendation evidence
is available and how it enters the shared node state.

\paragraph{Semantic neighborhoods.}
Visual and textual neighborhoods contribute their respective
semantic connections to one edge support. For $m\in\{v,t\}$,
projections $f_m$ map item features into a shared $d$-dimensional space.
We obtain normalized codes
$\mathbf z_i^m=\operatorname{norm}(f_m(\mathbf x_i^m))$ and
nonnegative affinities $g_{ij}^m=[(\mathbf z_i^m)^\top\mathbf z_j^m]_+$.
Let $\mathcal N_i^m$ contain the top-$K$ other items under modality $m$.
We retain a pair if either endpoint retrieves the other:
\begin{equation}
\label{eq:semantic-support}
B_{ij}^m=\mathbb I[j\in\mathcal N_i^m\lor i\in\mathcal N_j^m].
\end{equation}
The union $\Es=\{(i,j):B_{ij}^v+B_{ij}^t>0\}$ includes both
directions of each retained pair, allowing either modality to
establish a semantic connection. The retrieval budget $K$
applies to each modality before their supports are combined.

\paragraph{Behavior-informed modality weights.}
Behavioral agreement sets the relative contributions of visual
and textual relations. We form a co-occurrence
reference $o_{ij}=|\mathcal U_i\cap\mathcal U_j|/\sqrt{d_i d_j}$
and denote its top-$K$ neighbors by $\mathcal C_i$. The average overlap
$a_m=(NK)^{-1}\sum_i|\mathcal N_i^m\cap\mathcal C_i|$ measures this
agreement. A temperature $\tau_\alpha>0$ controls the concentration
of the resulting global modality weights:
\begin{equation}
\label{eq:modality-weight}
\alpha_m=
\frac{\exp(a_m/\tau_\alpha)}
{\sum_{m'\in\{v,t\}}\exp(a_{m'}/\tau_\alpha)}.
\end{equation}
These weights combine the supported affinities into raw semantic
edge weights, with $\beta>0$ setting their overall strength:
\begin{equation}
\label{eq:semantic-weight}
w^s_{i\rightarrow j}
=\beta\sum_{m\in\{v,t\}}\alpha_m B_{ij}^m g_{ij}^m.
\end{equation}
The raw semantic block is $S^0_{ji}=w^s_{i\rightarrow j}$.
We apply RGSC (Eq.~\ref{eq:budget-recovery}) to obtain
$\bar{\mathbf S}$ with the same incoming budgets.
Thus, $\alpha_m$ sets contributions across modalities,
while RGSC refines the allocation among individual neighbors.
The same $\alpha_m$ also combines modalities in the node anchors.

\paragraph{Multimodal node information.}
Semantic anchors provide a common profile for initializing
nodes and weighting their interactions. An item anchor combines
its modality codes using the graph's modality weights:
\begin{equation}
\label{eq:item-anchor}
\mathbf e_i^s=\operatorname{norm}\!\left(
\sum_{m\in\{v,t\}}\alpha_m\mathbf z_i^m\right).
\end{equation}
A user anchor summarizes its training history and incorporates
a group-level semantic prototype $\mathbf c_u$:
\begin{equation}
\label{eq:user-anchor}
\mathbf e_u^s=\operatorname{norm}\!\left(
\sum_{i\in\mathcal N_u}\pi_{ui}\mathbf e_i^s+\xi\mathbf c_u\right).
\end{equation}
The history weights satisfy $\pi_{ui}\propto(d_i+1)^{-1/2}$ and
sum to one over $\mathcal N_u$, reducing the relative influence
of popular items; $\xi$ controls the group contribution.
This local-plus-group initialization is related to
SG-URInit~\citep{xu2026well}.

ID embeddings and semantic anchors jointly initialize the
node state. Let $s_p=s_0$ for users and $s_p=s_0+\delta_p$
for items; then
\begin{equation}
\label{eq:node-initialization}
\mathbf h_p^{(0)}=\mathbf e_p^{id}+s_p\mathbf e_p^s.
\end{equation}
The learned scale $s_0$ controls semantic injection, with
item-specific adjustments $\delta_i$ given by a trainable
scale predictor. Its outputs are centered and bounded by
$|\delta_i|\leq\rho$.

\paragraph{Degree- and profile-aware interaction weights.}
Semantic profiles also determine the strength of observed
interactions, connecting node information to graph weighting. Let
$q_{ui}=[\cos(\mathbf e_u^s,\mathbf e_i^s)]_+$ measure this
profile agreement. An item-degree gate controls its influence:
\begin{equation}
\label{eq:degree-gate}
\eta_i=\sigma(a_c\log(d_i+1)+b_c).
\end{equation}
Here $a_c$ and $b_c$ are scalar hyperparameters controlling
degree sensitivity and gate offset, respectively.
The resulting interaction weight is
\begin{equation}
\label{eq:interaction-weight}
C_{ui}=R_{ui}\bigl(\eta_i+(1-\eta_i)q_{ui}\bigr).
\end{equation}
For an observed edge, $\eta_i$ sets a positive baseline weight,
and semantic agreement raises it toward one. Both directions
use the same weight, retaining observed connectivity while
modulating the strength of collaborative messages.

\paragraph{A shared graph learning context.}
The weighted collaborative and semantic relations jointly
define the neighborhoods of one graph over users and items:
\begin{equation}
\label{eq:ccmg-adjacency}
\mathbf A=
\begin{bmatrix}
\mathbf 0 & \mathbf C\\
\mathbf C^\top & \bar{\mathbf S}
\end{bmatrix}.
\end{equation}
Write $\mathbf A=\mathbf A_c+\mathbf A_s$, where $\mathbf A_c$
contains the user--item blocks and $\mathbf A_s$ contains the
item--item block.

Joint normalization makes message strength depend on both
relation families in this shared neighborhood. With
$\mathbf D=\operatorname{Diag}(\mathbf A\mathbf1)$,
the shared operator is
\begin{equation}
\label{eq:shared-normalization}
\widehat{\mathbf A}=\mathbf D^{-1/2}\mathbf A\mathbf D^{-1/2}.
\end{equation}
For item $j$, $D_{jj}=\sum_u C_{uj}+\mathcal B_j^s$ combines
its interaction strength and preserved semantic budget.
RGSC therefore adjusts semantic neighbor allocations without
changing the normalized collaborative coefficients.
CC-MG supplies CSCP with this
single operator and the initial node states, so both relation
families act on the same evolving representations.

\subsection{Collaborative--Semantic Compositional Propagation}
\label{sec:cscp}

We extend LightGCN~\citep{he2020lightgcn} to CC-MG through
multimodal initialization and collaborative--semantic propagation.
CSCP jointly updates shared node states with both relation families
at each layer.
This shared-state update enables cross-relation composition; the assembled
operator evaluates the updates in a single propagation stream.

\paragraph{Collaborative--semantic updates.}
Each layer updates users from their interacted items and items
from both users and semantic neighbors. Let
$\mathbf H_{\U}^{(l)}$ and $\mathbf H_{\I}^{(l)}$ denote the
user and item states at layer $l$, initialized by
Eq.~\ref{eq:node-initialization}. The user--item and item--item
blocks of $\widehat{\mathbf A}$ are denoted by
$\widehat{\mathbf C}$ and $\widehat{\mathbf S}$, respectively.
The user update is
\begin{equation}
\label{eq:cscp-user-update}
\mathbf H_{\U}^{(l+1)}
=\widehat{\mathbf C}\mathbf H_{\I}^{(l)}.
\end{equation}
The item update combines collaborative and semantic messages:
\begin{equation}
\label{eq:cscp-item-update}
\mathbf H_{\I}^{(l+1)}
=\widehat{\mathbf C}^{\top}\mathbf H_{\U}^{(l)}
+\widehat{\mathbf S}\mathbf H_{\I}^{(l)}.
\end{equation}
Here the first term aggregates user preferences along weighted
interactions, and the second aggregates item information along
calibrated semantic relations. Both use CC-MG's joint normalization
and the same preceding node state.

Shared-state updates carry recommendation evidence across relation types.
User information arriving at an item can reach semantic neighbors
in the next layer, while semantically enriched item states can
flow back to users. Deeper updates extend these collaborative--semantic
paths.

Within the semantic block, visual and textual relation contributions
act on the same item states and can compose across layers
to form V--T PECs.

\paragraph{Joint execution.}
A single sparse multiplication evaluates the user and item updates
together. Stacking their states into $\mathbf H^{(l)}$, we compute
\begin{equation}
\label{eq:cscp-propagation}
\mathbf H^{(l+1)}=\widehat{\mathbf A}\mathbf H^{(l)}.
\end{equation}
The assembled operator directly produces the next joint state,
combining collaborative and semantic messages in one sparse
operation per layer. It avoids separately materializing and
merging relation outputs.

\paragraph{Depth aggregation.}
The final representation combines initial node information with
evidence propagated over different path lengths. After $L$ layers,
we average the states:
\begin{equation}
\label{eq:depth-aggregation}
\bar{\mathbf H}=\frac{1}{L+1}\sum_{l=0}^{L}\mathbf H^{(l)}.
\end{equation}
The readout retains initial node information alongside evidence
from collaborative, semantic, and mixed paths at different depths.

\subsection{Prediction and Optimization}

The learned user and item representations determine preference
scores through their inner product:
\begin{equation}
\label{eq:recommendation-score}
\widehat y_{ui}=\bar{\mathbf h}_u^\top\bar{\mathbf h}_i,
\end{equation}
where $\bar{\mathbf h}_u$ and $\bar{\mathbf h}_i$ are the
corresponding rows of $\bar{\mathbf H}$.
We optimize the framework with Bayesian Personalized Ranking
\citep{rendle2012bpr}. Let $\mathcal D_{\mathrm{train}}$ contain
triplets $(u,i,j)$ with an observed item $i$ and a sampled
unobserved item $j$. Defining the score margin
$\Delta_{uij}=\widehat y_{ui}-\widehat y_{uj}$, the objective is
\begin{equation}
\label{eq:recommendation-loss}
\mathcal L=-\sum_{(u,i,j)\in\mathcal D_{\mathrm{train}}}
\log\sigma(\Delta_{uij})+\lambda\|\Theta\|_2^2.
\end{equation}
Here $\lambda$ controls regularization, and the trainable
parameters $\Theta$ comprise ID embeddings, $s_0$, and the
item-scale predictor. The modality projections, semantic anchors
(including group prototypes), and graph support and weights
remain fixed during recommendation training.

\section{Analysis of the Shared Graph Learning Context}
\label{sec:analysis-overview}

We fix the constructed graph and write
$\mathbf P_r=\mathbf D^{-1/2}\mathbf A_r\mathbf D^{-1/2}$ for $r\in\{c,s\}$, using the same joint degree matrix. Thus, $\mathbf P=\mathbf P_c+\mathbf P_s=\widehat{\mathbf A}$. The semantic-connectivity and PEC-composition results are provided in Appendix~\ref{app:semantic-analysis}; joint-execution equivalence and work and communication analysis in Appendix~\ref{sec:execution-analysis}; and calibration invariance in Appendix~\ref{app:calibration-invariance}.

\section{Experiments}
We address five research questions:
\textbf{RQ1:} How does \method compare with the baselines in ranking quality?
\textbf{RQ2:} How do semantic connectivity, graph weighting, normalization,
and C--S PEC composition affect recommendation?
\textbf{RQ3:} How does RGSC affect edge relevance and recommendation quality?
\textbf{RQ4:} Does joint execution reduce inference time at equivalent
outputs, and how does \method compare with MIG-GT in accuracy and
inference time?
\textbf{RQ5:} How do semantic-topology gains vary with user activity,
and how sensitive is performance to graph parameters?

\begin{table*}[!t]
\centering
\caption{Overall performance on three datasets. R and N denote Recall and NDCG; the best baseline is underlined.}
\label{tab:main}
\small
\setlength{\tabcolsep}{4.8pt}
\renewcommand{\arraystretch}{1.06}
\begin{tabular}{@{}l*{4}{c}!{\hspace{3pt}}*{4}{c}!{\hspace{3pt}}*{4}{c}@{}}
\toprule
\multirow{2}{*}{\textbf{Method}}
&\multicolumn{4}{c}{\textbf{Baby}}&\multicolumn{4}{c}{\textbf{Sports}}&\multicolumn{4}{c}{\textbf{Clothing}}\\
\cmidrule(lr){2-5}\cmidrule(lr){6-9}\cmidrule(lr){10-13}
&R@10&R@20&N@10&N@20&R@10&R@20&N@10&N@20&R@10&R@20&N@10&N@20\\
\midrule
MF & .0357&.0575&.0192&.0249&.0432&.0653&.0241&.0298&.0206&.0303&.0114&.0138\\
LightGCN & .0479&.0754&.0257&.0328&.0569&.0864&.0311&.0387&.0361&.0544&.0197&.0243\\
ApeGNN & .0501&.0775&.0267&.0338&.0608&.0892&.0333&.0407&.0378&.0538&.0204&.0244\\
MGDN & .0495&.0783&.0272&.0346&.0614&.0932&.0340&.0422&.0362&.0551&.0199&.0247\\
\midrule
VBPR & .0423&.0663&.0223&.0284&.0558&.0856&.0307&.0384&.0281&.0415&.0158&.0192\\
MMGCN & .0421&.0660&.0220&.0282&.0401&.0636&.0209&.0270&.0227&.0361&.0154&.0154\\
GRCN & .0532&.0824&.0282&.0358&.0599&.0919&.0330&.0413&.0421&.0657&.0224&.0284\\
DualGNN & .0513&.0803&.0278&.0352&.0588&.0899&.0324&.0404&.0452&.0675&.0242&.0298\\
SLMRec & .0521&.0772&.0289&.0354&.0663&.0990&.0365&.0450&.0442&.0659&.0241&.0296\\
LATTICE & .0547&.0850&.0292&.0370&.0620&.0953&.0335&.0421&.0492&.0733&.0268&.0330\\
FREEDOM & .0627&.0992&.0330&.0424&.0717&.1089&.0385&.0481&.0626&\underline{.0932}&.0338&.0416\\
MIG-GT & \underline{.0660}&.1015&\underline{.0357}&\underline{.0448}&\underline{.0738}&.1126&\underline{.0408}&\underline{.0508}&.0628&.0926&\underline{.0349}&\underline{.0425}\\
DHMRec & .0646&\underline{.1016}&.0343&.0436&\underline{.0738}&\underline{.1132}&.0401&.0502&\underline{.0629}&.0931&.0340&.0417\\
\midrule
\textbf{\method{} (Ours)} & \textbf{.0681}&\textbf{.1045}&\textbf{.0369}&\textbf{.0462}&\textbf{.0778}&\textbf{.1152}&\textbf{.0431}&\textbf{.0528}&\textbf{.0649}&\textbf{.0959}&\textbf{.0360}&\textbf{.0439}\\
\bottomrule
\end{tabular}
\end{table*}

\subsection{Experimental Setup}

\begin{table}[!htbp]
\centering
\caption{Statistics of the evaluation datasets.}
\label{tab:datasets}
\small
\begin{tabular}{lrrr}
\toprule
Dataset & Users & Items & Interactions\\
\midrule
Baby & 19,445 & 7,050 & 160,792\\
Sports & 35,598 & 18,357 & 296,337\\
Clothing & 39,387 & 23,033 & 278,677\\
\bottomrule
\end{tabular}
\end{table}

\subsubsection{Datasets and Evaluation}
We use Amazon Baby, Sports, and Clothing with MIG-GT's filtered data and pre-extracted visual/textual features \citep{hu2025modality} (Table~\ref{tab:datasets}). We follow its 8:1:1 split and full-ranking protocol, reporting Recall and NDCG at 10 and 20. Validation Recall@20 guides model selection and early stopping; graph construction, semantic anchors, and calibration use training interactions only.

\subsubsection{Baselines}
We compare two categories of baselines in Table~\ref{tab:main}.
The first uses only user--item interactions:
\textbf{MF}~\citep{koren2009matrix},
\textbf{LightGCN}~\citep{he2020lightgcn},
\textbf{ApeGNN}~\citep{zhang2023apegnn}, and
\textbf{MGDN}~\citep{hu2024mgdcf}.
The second combines interactions with item modality features:
\textbf{VBPR}~\citep{he2016vbpr},
\textbf{MMGCN}~\citep{wei2019mmgcn},
\textbf{GRCN}~\citep{wei2020graph},
\textbf{DualGNN}~\citep{wang2021dualgnn},
\textbf{SLMRec}~\citep{tao2022self},
\textbf{LATTICE}~\citep{zhang2021mining},
\textbf{FREEDOM}~\citep{zhou2023tale},
\textbf{MIG-GT}~\citep{hu2025modality}, and
\textbf{DHMRec}~\citep{zhan2026dhmrec}.
Reproduced and compatible published results follow the same benchmark
protocol.

\subsubsection{Implementation Details}
We use MIG-GT's fixed modality projections \citep{hu2025modality},
128-dimensional embeddings, a batch size of 8,000, and an RTX 3090.
Defaults are $K=10$, $\beta=(1.5,0.9,1.8)$ and depths $(3,2,2)$ for
Baby/Sports/Clothing; RGSC uses $F=5$, $\mu=10$, $\tau_r=0.1$,
and $\gamma=1$. Main results average five seeds; other experiments
use three unless specified. Ablations use paired seeds, fixed
context--target folds, matched initialization schemes, capacity,
and training budgets. Full uses the complete-model configuration
over the ablation seed set.

\subsubsection{Controlled Comparisons}
The graph-weighting and normalization controls are specified in Appendix~\ref{app:control-details}.
\emph{No S} removes semantic edges while retaining multimodal initialization.

The no-mixed and RGSC controls are detailed in Appendix~\ref{app:control-details}.

\subsubsection{Diagnostics and Timing}
Diagnostic definitions and timing operations are provided in Appendix~\ref{app:diagnostics-timing}.
Inference time covers all-node graph propagation at fixed weights
and parameters, excluding feature preparation, graph construction,
calibration, and training. R2 applies $\mathbf P_c$ and $\mathbf P_s$
separately to the same state; Joint applies their assembled sum, with equivalent
outputs (Proposition~\ref{prop:joint-execution}).  R2/Joint
share graph contents; all timing comparisons match hardware and
embedding dimension. MIG-GT accuracy comparisons use five-seed
main-result means.

\subsection{Overall Performance (RQ1)}
Table~\ref{tab:main} shows that \method achieves the highest reported
mean scores across all datasets and cutoffs. Improvements in both
Recall and NDCG indicate gains in retrieval coverage and rank-sensitive
quality. On Sports, relative gains are larger for shorter recommendation
lists. The gains over MIG-GT and DHMRec support the framework's ranking
quality against the compared modality-specific and staged designs.

\subsection{Evidence Organization in CC-MG (RQ2)}
\paragraph{Semantic relations beyond node information.}
Figure~\ref{fig:construction} shows that removing semantic topology
reduces both recommendation metrics on every dataset, most strongly
on Clothing. Semantic relations therefore contribute beyond the content
already supplied by node initialization. The contrast evaluates the
semantic relation block as a whole, including its effect on node degrees.

\paragraph{Graph weighting and normalization.}
In Figure~\ref{fig:ccmg-design}, Full obtains the highest mean scores
across all datasets. Separate normalization yields the largest decrease,
followed by uniform interaction weights; equal modality weights have
a smaller but consistent effect on both metrics. These results support
degree- and profile-aware interaction weighting, behavior-informed
modality weighting, and joint normalization within CC-MG.
The stronger normalization effect highlights the value of letting both
relation families determine node degrees and message strengths.

\begin{figure}[!htbp]
  \centering
  \includegraphics[width=\columnwidth]{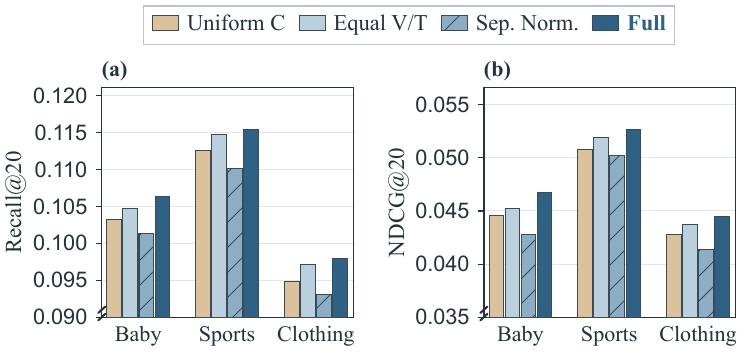}
  \caption{CC-MG construction ablation. (a) Recall@20; (b) NDCG@20. Axes are truncated.}
  \label{fig:ccmg-design}
  \Description{Two grouped-bar panels compare Recall@20 and NDCG@20
  across graph-construction variants on Baby, Sports, and Clothing.}
\end{figure}

\begin{figure}[!htbp]
  \centering
  \includegraphics[width=\columnwidth]{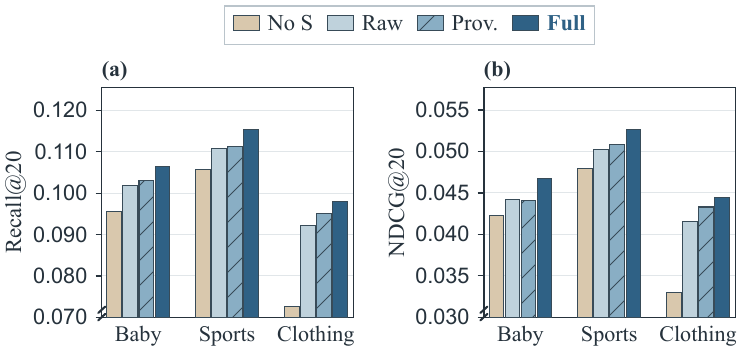}
  \caption{Semantic topology and RGSC ablation. (a) Recall@20; (b) NDCG@20. Axes are truncated.}
  \label{fig:construction}
  \Description{Two grouped-bar panels compare No S, Raw, provisional tilting and Full across three datasets.}
\end{figure}

\paragraph{Cross-relation composition.}
Figure~\ref{fig:mixed-path-ablation} shows consistent Recall@20 and
NDCG@20 gains over no-mixed on all datasets. These gains support the
recommendation value of C--S PEC formation beyond semantic connectivity
and relation-pure propagation. Together with the preliminary diagnostics
in Figure~\ref{fig:fragmentation-evidence}, they provide complementary
evidence for useful compositions in the shared graph learning context.

\begin{figure}[!htbp]
  \centering
  \includegraphics[width=.90\columnwidth]{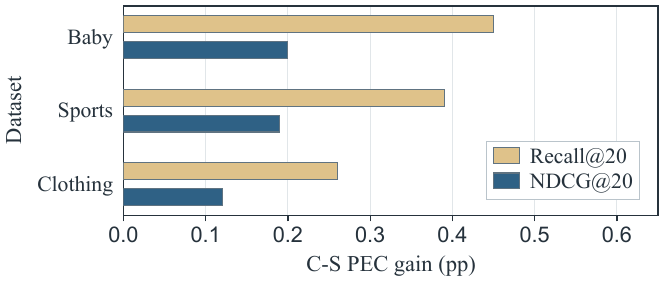}
  \caption{C--S PEC composition. Recall@20 and NDCG@20 gains over no-mixed, in percentage points.}
  \label{fig:mixed-path-ablation}
  \Description{Six paired horizontal bars report gains from cross-relation composition.}
\end{figure}

\subsection{Behavioral Calibration within CC-MG (RQ3)}
Figure~\ref{fig:construction} shows that Full outperforms Raw in both
recommendation metrics on every dataset. By
Eq.~\ref{eq:calibration-invariance}, these gains support behavioral
neighbor reallocation within preserved semantic budgets.

Prov.\ improves Recall@20 over Raw on all datasets and NDCG@20 on
Sports and Clothing. On Baby, NDCG@20 decreases slightly despite higher
Recall@20, so behavioral tilting alone does not improve both metrics
uniformly. Full outperforms Prov.\ in both metrics on every dataset,
supporting budget recovery under the matched configuration. This
contrast includes its effect on semantic and collaborative normalization.

The edge-relevance diagnostics and their relationship to downstream recommendation are reported in Appendix~\ref{app:edge-relevance}.

\subsection{Joint Execution Efficiency (RQ4)}
Figure~\ref{fig:app-runtime-scaling} shows that inference time grows
with directed-edge count, while Joint remains faster than R2 across
the displayed range on all datasets. With equivalent outputs and message
work, the gap is consistent with reduced execution overhead, as analyzed
in Appendix~\ref{sec:execution-analysis}. Timing alone does not isolate
the contribution of each overhead source.

\begin{figure}[!htbp]
  \centering
  \includegraphics[width=.94\columnwidth]{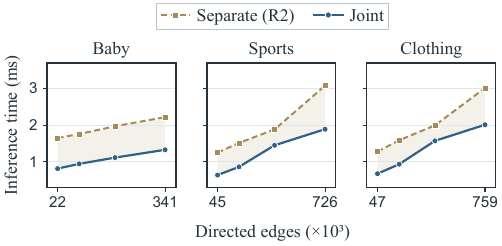}
  \caption{CSCP joint execution. Shading marks the inference-time gap between R2 and Joint. Directed-edge ticks are in thousands.}
  \label{fig:app-runtime-scaling}
  \Description{Three dataset panels plot inference time against
  directed-edge count for Separate (R2) and Joint with equivalent outputs.}
\end{figure}

Figure~\ref{fig:app-accuracy-efficiency} shows that \method combines
Recall@20 gains of 0.30, 0.26, and 0.33 percentage points with inference
speedups of $3.25\times$, $5.65\times$, and $7.88\times$ on Baby,
Sports, and Clothing. This complete-model comparison demonstrates a
combined accuracy and runtime benefit, complementing the R2/Joint
control. Since the models differ in other components, it does not
attribute the entire speedup to joint execution.

\begin{figure}[!htbp]
  \centering
  \includegraphics[width=.90\columnwidth]{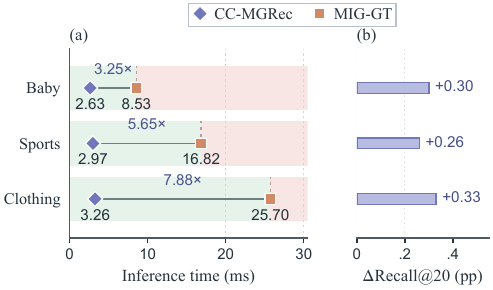}
  \caption{Complete-model comparison with MIG-GT.
  (a) Inference times; green/red bands mark times below/above each MIG-GT
  reference, and labels above the pairs give speedup
  $t_{\mathrm{GT}}/t_{\mathrm{CC}}$.
  (b) Recall@20 gains, $100(R_{\mathrm{CC}}-R_{\mathrm{GT}})$ pp.}
  \label{fig:app-accuracy-efficiency}
  \Description{Three aligned dataset rows pair CC-MGRec and MIG-GT
  inference times, with pale backgrounds split at each MIG-GT reference.
  Time ratios appear above the pairs. Separate horizontal bars show
  Recall@20 gains over MIG-GT.}
\end{figure}

Across the controls, semantic connectivity, C--S PEC composition, and
calibration support ranking quality, while joint execution reduces
inference overhead. These contrasts assess complementary roles under
matched comparisons; their gains cannot be added as independent
component contributions.

\subsection{User Activity and Parameter Sensitivity (RQ5)}
Figure~\ref{fig:app-sparsity} shows positive semantic-topology gains
in every activity group. Tail users have the largest gain within each
dataset, while Clothing retains a large benefit for Head users.
Semantic relations thus help across activity levels, with the strongest
benefit for the least active users.

\begin{figure}[!htbp]
  \centering
  \includegraphics[width=.90\columnwidth]{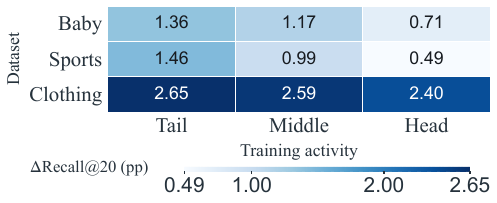}
  \caption{Semantic topology gains by training activity. Cells show Recall@20 gains over No S, in percentage points.}
  \label{fig:app-sparsity}
  \Description{Nine cells retain positive gains across Tail, Middle and Head user groups.}
\end{figure}

The parameter-sensitivity results and analysis are provided in Appendix~\ref{app:parameter-sensitivity}.

\section{Conclusion}
\method organizes SBE within a shared graph learning context for
multimodal recommendation, supporting semantic connectivity and
C--S and V--T PEC composition. CC-MG jointly normalizes collaborative
and semantic relations; RGSC calibrates semantic neighbors while
preserving normalized collaborative propagation coefficients; and
CSCP updates shared node states across layers. Experiments on three
Amazon datasets demonstrate improved ranking over the compared
baselines. Controlled comparisons support semantic connectivity
beyond node information, C--S PEC composition beyond relation-pure
propagation, and behavioral calibration. The output-equivalent
R2/Joint control shows that joint execution reduces inference time
while preserving propagation outputs. Timing covers fixed-weight
graph propagation, excluding feature preparation, graph construction,
calibration, and training.


\clearpage
\bibliographystyle{ACM-Reference-Format}
\bibliography{reference}

\clearpage
\appendix
\section{Theoretical Analysis and Proofs}
\label{app:theory}
\label{sec:properties}

The analysis addresses the two objectives of our motivation:
aggregating recommendation evidence through semantic connectivity
and cross-relation composition, and executing their propagation
with low overhead. We fix the constructed graph and write
$\mathbf P_r=\mathbf D^{-1/2}\mathbf A_r\mathbf D^{-1/2}$
for $r\in\{c,s\}$, using the same joint degree matrix.
Thus, $\mathbf P=\mathbf P_c+\mathbf P_s=\widehat{\mathbf A}$.
The operators are nonnegative and may be asymmetric;
symmetry and commutativity are not required.
Inverse degree factors are zero at zero-degree nodes.
Let $n=|\U|+|\I|$, let $d$ be the state width, and let $L\geq1$
be the propagation depth; empty sums below are zero.

\subsection{Semantic Connectivity and Cross-Relation Composition}
\label{app:semantic-analysis}

The final readout retains the propagation routes available
at every depth. By Eqs.~\ref{eq:cscp-propagation}
and~\ref{eq:depth-aggregation}, $\bar{\mathbf H}=\mathbf K_L\mathbf H^{(0)}$,
where
\begin{equation}
\label{eq:propagation-readout}
\mathbf K_L=\frac{1}{L+1}\sum_{\ell=0}^{L}\mathbf P^\ell.
\end{equation}
The following comparisons hold the initial node states and
joint normalization fixed to isolate these routes.
Ordered adjacency products represent composite relations
\citep[Section~3.1]{yun2019graph}. For a relation sequence
$\omega=(r_1,\ldots,r_\ell)$, define
$\mathbf P_\omega=\mathbf P_{r_\ell}\cdots\mathbf P_{r_1}$,
in propagation order. Let $\Omega_\ell$ contain sequences
that use both $c$ and $s$.

\begin{proposition}[Decomposition of propagated evidence]
\label{prop:evidence-decomposition}
The shared-state readout decomposes as
\begin{equation}
\label{eq:evidence-decomposition}
\mathbf K_L=\mathbf K_c+\mathbf K_s+\mathbf K_{\mathrm{mix}},
\end{equation}
where
\begin{equation}
\label{eq:collaborative-readout}
\mathbf K_c=\frac{1}{L+1}\left(\mathbf I_n+
\sum_{\ell=1}^{L}\mathbf P_c^\ell\right),
\end{equation}
\begin{equation}
\label{eq:semantic-readout}
\mathbf K_s=\frac{1}{L+1}\sum_{\ell=1}^{L}\mathbf P_s^\ell,
\end{equation}
and
\begin{equation}
\label{eq:mixed-readout}
\mathbf K_{\mathrm{mix}}=\frac{1}{L+1}
\sum_{\ell=2}^{L}\sum_{\omega\in\Omega_\ell}\mathbf P_\omega.
\end{equation}
These terms respectively retain collaborative-only propagation,
semantic-only propagation, and cross-relation composition.
All three operators are entrywise nonnegative.
\end{proposition}

\begin{proof}
Repeated left multiplication by $\mathbf P_c+\mathbf P_s$
gives
\begin{equation}
\label{eq:ordered-expansion}
\mathbf P^\ell=\sum_{\omega\in\{c,s\}^{\ell}}\mathbf P_\omega.
\end{equation}
For $\ell=1$, this is the definition of $\mathbf P$.
If it holds at depth $\ell$, left multiplication appends
either $c$ or $s$ to each sequence, producing every ordered
sequence of length $\ell+1$ exactly once. This establishes
the expansion by induction without exchanging matrix factors.
Partitioning the sequences into all-$c$, all-$s$, and mixed
sequences, then averaging depths zero through $L$, yields
Eq.~\ref{eq:evidence-decomposition}. The identity is counted
once at depth zero. Nonnegativity follows from nonnegative
matrix products and sums.
\end{proof}

\begin{corollary}[Semantic connectivity beyond node information]
\label{cor:semantic-connectivity}
For the same initial states, a collaborative-only operator
$\mathbf K_c$ omits the contributions $\mathbf K_s+\mathbf K_{\mathrm{mix}}$.
For any positive normalized semantic edge $i\rightarrow j$ with $i\ne j$,
\begin{equation}
\label{eq:semantic-dependence}
(\mathbf K_L-\mathbf K_c)_{ji}
\geq\frac{(\mathbf P_s)_{ji}}{L+1}>0.
\end{equation}
If no collaborative-only walk of length at most $L$ connects
$i$ to $j$, then $(\mathbf K_c)_{ji}=0$ but $(\mathbf K_L)_{ji}>0$.
\end{corollary}

\begin{proof}
The decomposition gives
$\mathbf K_L-\mathbf K_c=\mathbf K_s+\mathbf K_{\mathrm{mix}}$.
Its terms are nonnegative, and the depth-one semantic term
contains $\mathbf P_s/(L+1)$, proving the bound.
Each entry of $\mathbf P_c^\ell$ sums the weights of
collaborative-only walks of length $\ell$. With no such walk
and distinct endpoints, every contribution to $(\mathbf K_c)_{ji}$
is zero.
\end{proof}

Semantic connectivity therefore supplies explicit neighbor
aggregation in addition to the information carried by the
initial node states. In CC-MG, its first-layer item contribution
is $\widehat{\mathbf S}\mathbf H_{\I}^{(0)}$.
The result characterizes propagation of given node inputs;
their initialization may already contain semantic information.

\begin{corollary}[Strict contribution of C--S PECs]
\label{cor:mixed-dependence}
The relation-pure readout used by no-mixed is
$\mathbf K_L^{\mathrm{pure}}=\mathbf K_c+\mathbf K_s$.
Suppose a mixed walk
$\pi=(v_0=q,\ldots,v_t=p)$ has $t\leq L$ and positive
normalized weights
$\widehat w_a=(\mathbf P_{r_a})_{v_a v_{a-1}}$,
where its relation sequence uses both $c$ and $s$. Then
\begin{equation}
\label{eq:mixed-dependence}
(\mathbf K_L-\mathbf K_L^{\mathrm{pure}})_{pq}
\geq\frac{1}{L+1}\prod_{a=1}^{t}\widehat w_a>0.
\end{equation}
For distinct $p,q$, if neither a collaborative-only nor a
semantic-only walk of length at most $L$ connects them,
$(\mathbf K_L^{\mathrm{pure}})_{pq}=0$ and $(\mathbf K_L)_{pq}>0$.
\end{corollary}

\begin{proof}
By Proposition~\ref{prop:evidence-decomposition}, the difference
is $\mathbf K_{\mathrm{mix}}$.
For the sequence of $\pi$, the $(p,q)$ entry of
$\mathbf P_{r_t}\cdots\mathbf P_{r_1}$ sums products of weights
over matching walks. One summand is
$\prod_{a=1}^{t}\widehat w_a>0$. This product appears in the
depth-$t$ contribution to $\mathbf K_{\mathrm{mix}}$, and all
remaining terms are nonnegative, establishing the bound.
Absence of pure walks makes both pure-operator entries zero.
\end{proof}

A positive readout coefficient establishes a propagation
dependence between the two nodes. Holding the graph fixed,
the sensitivity to an initial node state is
\begin{equation}
\label{eq:node-state-dependence}
\frac{\partial\bar{\mathbf h}_p}{\partial\mathbf h_q^{(0)}}
=(\mathbf K_L)_{pq}\mathbf I_d.
\end{equation}
For example, $\mathbf P_c\mathbf P_s$ lets an item's state
reach a user through a semantic neighbor and an observed
interaction; $\mathbf P_s\mathbf P_c$ lets user information
continue through semantic neighbors. The bounds explain
when shared-state propagation adds such dependencies to the
final readout. These dependencies concern propagated
representations; the final inner product can also combine
information from both pure branches.

For modality composition, RGSC applies the same edge-wise tilt and
destination-wise recovery factor to both terms of the raw semantic weight.
Thus $\bar{\mathbf S}=\bar{\mathbf S}_v+\bar{\mathbf S}_t$ and,
under the shared degree matrix, $\mathbf P_s=\mathbf P_v+\mathbf P_t$.
Expanding semantic products in the same propagation order yields V--T PECs
when both modality operators contribute positive edges.

\subsection{Efficiency through Joint Execution}
\label{sec:execution-analysis}

We distinguish edge--feature message work from relation-output
communication, following work-based GNN analysis and aggregation
reorganization \citep{besta2024parallel,jia2020redundancy}.

\paragraph{Equivalent shared-neighborhood aggregation.}
Let $\mathcal N_r(p)=\{q:(\mathbf P_r)_{pq}>0\}$ for $r\in\{c,s\}$,
$\mathcal N(p)=\mathcal N_c(p)\cup\mathcal N_s(p)$, and
$P_{pq}=(\mathbf P_c)_{pq}+(\mathbf P_s)_{pq}$. Joint computes
\begin{equation}
\label{eq:shared-neighborhood-aggregation}
\mathbf h_p^{(\ell+1)}
=\sum_{q\in\mathcal N(p)}P_{pq}\mathbf h_q^{(\ell)}.
\end{equation}

\begin{proposition}[Equivalent shared-neighborhood aggregation]
\label{prop:joint-execution}
For the same fixed graph and initialization, Joint and R2
produce identical layer states and readouts in exact arithmetic.
With the same objective and parameterization, their parameter
gradients coincide wherever the derivatives exist.
\end{proposition}

\begin{proof}
Grouping relation contributions by neighbor gives
Eq.~\ref{eq:shared-neighborhood-aggregation}, equivalently
\begin{equation}
\label{eq:joint-equivalence}
\mathbf P_c\mathbf H^{(\ell)}
+\mathbf P_s\mathbf H^{(\ell)}
=\mathbf P\mathbf H^{(\ell)}.
\end{equation}
Induction gives identical layer states and the common readout
$\mathbf K_L\mathbf H^{(0)}$. For an output gradient $\mathbf Z$,
\begin{equation}
\label{eq:joint-adjoint}
\mathbf P_c^\top\mathbf Z+\mathbf P_s^\top\mathbf Z
=\mathbf P^\top\mathbf Z.
\end{equation}
The chain rule through the common readout and initialization
establishes gradient equivalence.
\end{proof}

\paragraph{Propagation work.}
Let $m_c=\operatorname{nnz}(\mathbf P_c)$ and
$m_s=\operatorname{nnz}(\mathbf P_s)$ count directed nonzero edges.
Count one message evaluation per edge and feature coordinate,
excluding graph preparation, initialization, and readout.

\begin{proposition}[Propagation work and consolidation condition]
\label{prop:graph-message-work}
CC-MG's $L$-layer, $d$-dimensional message work is
\begin{equation}
\label{eq:joint-message-work}
\mathcal W_{\mathrm{Joint}}=Ld(m_c+m_s).
\end{equation}
For $B$ streams of width $d$, with depths $L_b$ and edge counts $m_b$,
\begin{equation}
\label{eq:multistream-message-work}
\mathcal W_{\mathrm{streams}}=d\sum_{b=1}^{B}L_bm_b.
\end{equation}
When $\mathcal W_{\mathrm{streams}}>0$, the fractional reduction is
\begin{equation}
\label{eq:message-work-reduction}
\rho_{\mathrm{msg}}
=1-\frac{L(m_c+m_s)}{\sum_{b=1}^{B}L_bm_b}.
\end{equation}
Reduction occurs exactly when $L(m_c+m_s)<\sum_bL_bm_b$.
\end{proposition}

\begin{proof}
Disjoint user--item and item--item supports give
$\sum_p|\mathcal N(p)|=m_c+m_s$.
Each layer evaluates $d$ coordinates per nonzero edge.
Multiplying by depth and summing over streams gives the two
work counts; their relative difference gives $\rho_{\mathrm{msg}}$.
\end{proof}

Two modalities retrieve at most $2K|\I|$ directed neighbors
in total. Symmetrization at most doubles this count, and
calibration preserves support; hence
\begin{equation}
\label{eq:semantic-message-budget}
m_s\leq4K|\I|.
\end{equation}
For $B$ equal-depth streams over the interaction support
$m_c>0$, let $\tau=m_s/m_c$. Then
\begin{equation}
\label{eq:equal-depth-work-reduction}
\rho_{\mathrm{msg}}=1-\frac{1+\tau}{B}.
\end{equation}
It is positive exactly when $\tau<B-1$.
Separate streams may use distinct states and depths
\citep{hu2025modality}; the work comparison does not assert
equivalent outputs or redundant messages.

\paragraph{Relation-output communication.}
R2 and Joint have the same message work.
For output accesses, count only relation-output stores and loads
and final next-state stores at operator boundaries.

\begin{corollary}[Eliminated relation-output communication]
\label{cor:output-communication}
If R2 stores two dense $n\times d$ relation outputs and
reads them for separate fusion, then over $L$ layers
\begin{equation}
\label{eq:execution-access}
Q_{\mathrm{R2}}^{\mathrm{out}}=5Lnd,\qquad
Q_{\mathrm{Joint}}^{\mathrm{out}}=Lnd.
\end{equation}
Joint eliminates
\begin{equation}
\label{eq:execution-savings}
\Delta Q^{\mathrm{out}}=4Lnd.
\end{equation}
\end{corollary}

\begin{proof}
Per layer, R2 incurs $2nd$ intermediate stores, $2nd$ fusion
reads, and $nd$ next-state stores; Joint incurs only the last
$nd$ stores. Multiplication by $L$ gives the counts.
\end{proof}

Joint removes 80\% of these output accesses and reduces
sparse-operator calls from $2L$ to $L$.
This count excludes sparse inputs, aggregation-internal accesses,
and readout. Materialization and relation-wise execution overhead
are also identified in GNN systems \citep{wu2024hector,han2025tlv};
physical latency additionally depends on kernels, caching,
and scheduling.

\subsection{Calibration Invariance}
\label{app:calibration-invariance}
Budget preservation makes RGSC compatible with this shared normalization while refining semantic neighbor allocations.

\begin{lemma}[Calibration invariance]
\label{lem:calibration-invariance}
For fixed $\mathbf C$, preserving every incoming semantic
budget leaves the joint degree matrix and normalized
collaborative operator unchanged:
\begin{equation}
\label{eq:calibration-invariance}
\mathbf D=\mathbf D^0,\qquad \mathbf P_c=\mathbf P_c^0.
\end{equation}
Here the superscript $0$ denotes the raw semantic graph.
\end{lemma}

\begin{proof}
RGSC ensures $\bar{\mathbf S}\mathbf1=\mathbf S^0\mathbf1$.
User degrees depend only on $\mathbf C$, and item degrees
are $\sum_u C_{uj}+\mathcal B_j^s$, with both terms unchanged.
Hence $\mathbf D=\mathbf D^0$. The collaborative adjacency
is also unchanged, so using this degree matrix on both sides
gives $\mathbf P_c=\mathbf P_c^0$.
\end{proof}

This property fixes collaborative propagation coefficients;
the reallocated semantic messages can still change learned
states and recommendation scores.

\section{Supplementary Experimental Settings}
\label{app:experimental-details}

\subsection{Graph-Construction Controls}
\label{app:control-details}
\emph{Uniform C} makes interaction weights constant at matched total
weight; \emph{Equal V/T} uses equal modality coefficients while
preserving incoming semantic budgets before RGSC and node-information inputs;
\emph{Sep.\ Norm.} independently normalizes relation blocks, with
validation-selected relative strength and matched total operator weight.
These controls retain support, RGSC, depth, and shared-state updates.
The retrained \emph{no-mixed} control uses the relation-pure readout
(Corollary~\ref{cor:mixed-dependence}) with the same $\mathbf P_c$ and
$\mathbf P_s$. It prevents C--S PEC formation while retaining V--T PECs
in the semantic branch.
\emph{Raw} disables calibration ($\gamma=0$); \emph{Prov.} omits budget
recovery; \emph{Full} includes both RGSC stages. Candidate support,
global modality weights, and initialization schemes are retained.
Raw and Full share $\mathbf D$ and $\mathbf P_c$
(Lemma~\ref{lem:calibration-invariance}).

\subsection{Diagnostics and Timing Operations}
\label{app:diagnostics-timing}
Figure~\ref{fig:fragmentation-evidence}(a) matches interaction degree,
semantic support, and incoming semantic mass; (b) intervenes in fixed
models, retaining graph support, initialization, first-order information,
and single-modality paths. Edge relevance evaluates fixed candidate
edges pooled across destinations using AUC, average precision (AP),
and positive lift; held-out labels are used only for evaluation.
Activity groups use training activity, and sensitivity grids use Baby.
Timings use CUDA synchronization and repeated measurements after warm-up.

\section{Supplementary Experimental Results}
\label{app:extra-results}

\subsection{Edge Relevance and Downstream Recommendation}
\label{app:edge-relevance}

In Figure~\ref{fig:rarc-edge-validation}, Full improves AP and positive
lift over Raw on all datasets. Its AUC increases on Baby and Sports
but decreases on Clothing; Prov.\ exceeds Full on all displayed edge
metrics. Budget recovery preserves local pre-normalization neighbor order but changes
score scales across destinations, which can alter pooled rankings.
Despite lower pooled edge-ranking metrics, Full yields better
recommendation performance than Prov.\ on every dataset. Edge ranking and
message allocation assess different uses of the weights: better pooled
edge discrimination does not necessarily yield better downstream recommendation.
The downstream results favor reallocation within preserved budgets.

\begin{figure}[!htbp]
  \centering
  \includegraphics[width=.96\columnwidth]{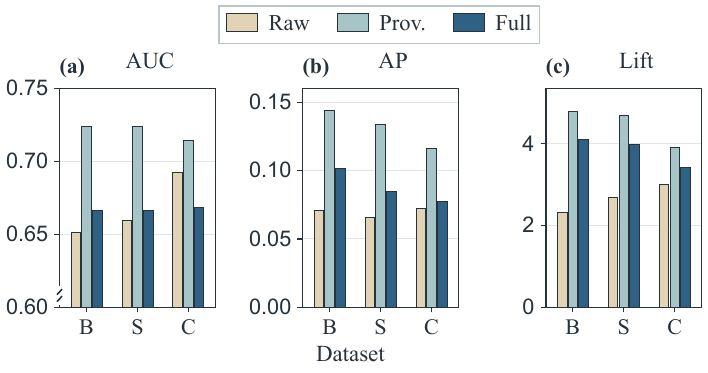}
  \caption{RGSC edge relevance. (a) AUC; (b) AP; (c) positive lift. B/S/C denote Baby/Sports/Clothing. AUC is truncated at 0.60.}
  \label{fig:rarc-edge-validation}
  \Description{Three panels compare pooled held-out edge AUC, AP,
  and positive lift for Raw, Prov., and Full on Baby, Sports, and Clothing.}
\end{figure}

\subsection{Parameter Sensitivity}
\label{app:parameter-sensitivity}

Figure~\ref{fig:app-semantic-sensitivity}(a,b) shows that semantic
strength depends on neighborhood size: stronger weighting remains
competitive for smaller neighborhoods but reduces performance in the
densest setting. Neighborhood density and message strength should
therefore be tuned jointly.

In Figure~\ref{fig:app-calibration-sensitivity}(c,d), Raw yields
Recall@20/NDCG@20 of .1018/.0442 at every temperature because calibration
is disabled. Among positive magnitudes, Recall and NDCG favor different
strengths; stronger calibration does not improve both uniformly.

\begin{figure}[!htbp]
  \centering
  \includegraphics[width=.96\columnwidth]{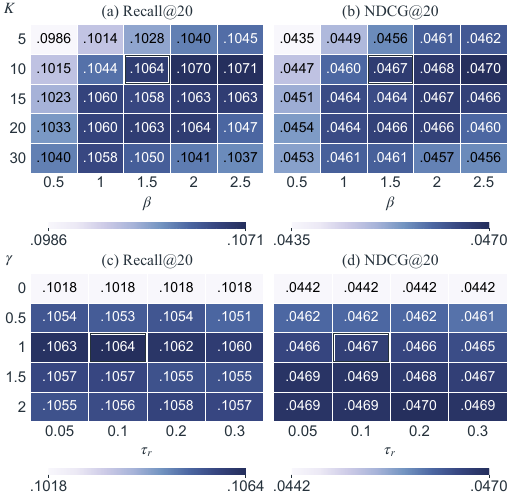}
  \caption{Parameter sensitivity on Baby. (a,b) CC-MG Recall@20/NDCG@20 across $K$ and $\beta$; (c,d) RGSC Recall@20/NDCG@20 across $\gamma$ and $\tau_r$. Outlines mark validation-selected defaults.}
  \label{fig:app-semantic-sensitivity}
  \label{fig:app-calibration-sensitivity}
  \Description{Four heatmaps show Recall@20 and NDCG@20 across
  graph-construction and calibration parameters; outlined cells
  indicate validation defaults.}
\end{figure}

\end{document}